\documentclass{article}
\usepackage{spconf,amsmath,graphicx}
\usepackage{subfigure}
\usepackage{colortbl}
\usepackage{bm}
\usepackage{stfloats}
\usepackage{graphicx}  
\usepackage{epstopdf}
\usepackage{url}       
\usepackage{amsmath}   
\usepackage{amsfonts,amssymb}
\usepackage{cases}
\usepackage{algorithm}
\usepackage{multirow}
\usepackage{algorithmic}
\usepackage{microtype}
\usepackage{setspace}
 \usepackage{enumerate}
\usepackage{amsthm}
\usepackage[style=ieee, doi=false, isbn=false, url=false]{biblatex}

\newtheorem{proposition}{{Proposition}}
\newtheorem{remark}{{Remark}}
\newcommand{\ls}[1]
    {\dimen0=\fontdimen6\the\font
     \lineskip=#1\dimen0
     \advance\lineskip.5\fontdimen5\the\font
     \advance\lineskip-\dimen0
     \lineskiplimit=.9\lineskip
     \baselineskip=\lineskip
     \advance\baselineskip\dimen0
     \normallineskip\lineskip
     \normallineskiplimit\lineskiplimit
     \normalbaselineskip\baselineskip
     \ignorespaces
    }

\definecolor{purple}{RGB}{128,0,128}

\title{Atomic Norm Based Localization and Orientation Estimation for Millimeter-Wave MIMO OFDM Systems \vspace{-10pt}}

\name{Jianxiu Li, Maxime Ferreira Da Costa, and Urbashi Mitra\thanks{ This work has been funded in part by one or more of the following: Cisco Foundation 1980393, ONR N00014-15-1-2550, ONR 503400-78050, NSF CCF-1410009, NSF CCF-1817200, NSF CCF-2008927, Swedish Research Council 2018-04359, ARO W911NF1910269, and DOE DE-SC0021417.}
\thanks{Authors' emails: \texttt{\{{jianxiul, mferreira, ubli}\}@usc.edu}}\vspace{-10pt}}
\address{Ming Hsieh Department of Electrical Engineering, Viterbi School of Engineering\\
University of Southern California, Los Angeles, CA, USA
\vspace{-15pt}}

\let\hat\widehat
\let\tilde\widetilde

\begin{document}
%
\maketitle
\begin{abstract} \vspace{-5pt}
Herein, an atomic norm based method for accurately estimating the location and orientation of a target from millimeter-wave multi-input-multi-output (MIMO) orthogonal frequency-division multiplexing (OFDM) signals is presented. A novel virtual channel matrix is introduced and an algorithm to extract localization-relevant channel parameters from its atomic norm decomposition is designed. Then, based on the extended invariance principle, a weighted least squares problem is proposed to accurately recover the location and orientation using both line-of-sight and non-line-of-sight channel information. Numerical results highlight performance improvements over a prior method and the resultant performance nearly achieves the Cram\'{e}r-Rao lower bound.
\end{abstract}

\vspace{-5pt}

\vspace{-10pt}
\section{Introduction}
\label{sec:intro}
\vspace{-5pt}
Millimeter-wave (mmWave) communications are widely adopted in the fifth generation multi-input-multi-output (MIMO) systems due to the massive available bandwidth and tiny wavelength \cite{Rappaport}. While the path loss they incur is a core challenge, it can be mitigated by beamforming \cite{Hur}. Due to the limited resulting number of paths, mmWave MIMO signaling is of interest for localization \cite{Hua,Saloranta,Shahmansoori}. In \cite{Saloranta,Shahmansoori},  localization based on classic compressed sensing is pursued by exploiting multipath sparsity.
However, the performance is limited by quantization error and grid resolution.  In \cite{Shahmansoori}, a space-alternating generalized expectation maximization (SAGE) algorithm is proposed to refine the channel estimates for localization, initialized with the channel parameters that are coarsely estimated via a modified distributed compressed sensing simultaneous orthogonal matching pursuit (DCS-SOMP) scheme \cite{Duarte}. However, this method suffers from local minima when the signal-to-noise ratio (SNR) is low or when the initialization is not sufficiently accurate.

On the other hand, atomic norm minimization \cite{chi2020harnessing,candes2014towards,TangG} (ANM, \emph{a.k.a.} total variation minimization) has emerged as a convex optimization framework for estimating continuous valued parameters without relying on discretization. ANM is robust to noise \cite{bhaskar2013atomic,da2020stable}. It has previously been employed  for the purpose of localization \cite{Wu1,Wu2,TangW}, but without multipath considerations. In \cite{Tsai}, ANM is used for channel estimation, but cannot be used for localization as time-of-arrival is not considered in the model. In this paper, we design an ANM based approach for high-accuracy localization and orientation estimation using  mmWave MIMO orthogonal frequency-division multiplexing (OFDM) signaling.

The main contributions of this paper are:\vspace{-6pt}
\begin{enumerate}[1)]
\item A novel {\it virtual channel matrix} is designed for mmWave MIMO OFDM multi-path channels. \vspace{-6pt}
\item A multi-dimensional atomic norm based channel estimator is proposed for positioning purposes, where the structure of the proposed virtual channel matrix is explicitly exploited to simultaneously estimate TOAs, angle-of-arrivals (AOA), and angle-of-departures (AOD) with super-resolution; \vspace{-6pt}
\item To accurately recover location and orientation, a weighted least squares scheme is proposed based on the extended invariance principle (EXIP) \cite{Stoica}, where the designed weight matrix is compatible with the ANM channel estimator;\vspace{-6pt}
\item Numerical comparisons to the DCS-SOMP based method \cite{Shahmansoori} show that the proposed scheme offers more than $7$dB gain with respect to the root-mean-square error (RMSE) of estimation when a small number of antennas are employed.  Furthermore,  the proposed method nearly achieves the
 Cram\'{e}r-Rao lower bound (CRLB) \cite{Shahmansoori} in many cases.
\end{enumerate}
\vspace{-20pt}
\section{Signal Model}\label{sec:signal} \vspace{-7pt}
We adopt the narrowband channel model of \cite{Shahmansoori}, where a single base station (BS) is equipped with $N_t$ antennas and a target has $N_r$ antennas. The locations of the BS and the target are denoted by $\boldsymbol{q}=\left[q_{x}, q_{y}\right]^{\mathrm{T}} \in \mathbb{R}^{2}$ and $\boldsymbol{p}=\left[q_{x}, q_{y}\right]^{\mathrm{T}} \in \mathbb{R}^{2}$, respectively, where $\boldsymbol{q}$ is known while $\boldsymbol{p}$ is to be estimated. In addition, there is an unknown orientation of the target's antenna array, denoted by $\theta_o$. Assume that one line-of-sight (LOS) path and $K$ non-line-of-sight (NLOS) paths exist in the mmWave MIMO OFDM channel and the total number of the paths is known. The $k$-th NLOS path is produced by a scatterer at an unknown location $\boldsymbol{s}_k=\left[s_{k,x}, s_{k,y}\right]^{\mathrm{T}} \in \mathbb{R}^{2}$.
\begin{figure}[t]
    \centering
    \includegraphics[width=0.4\textwidth]{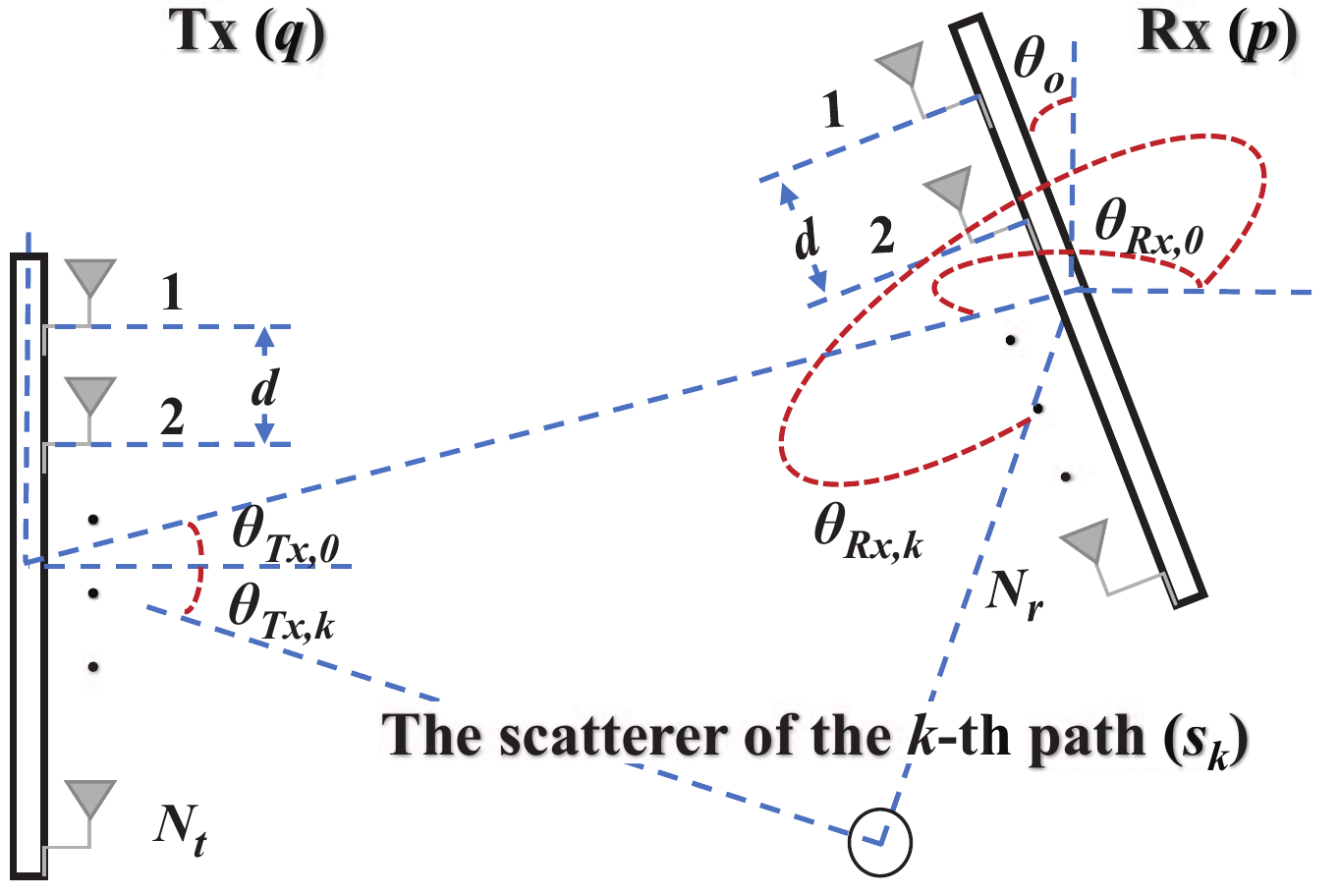} \vspace{-10pt}
    \caption{~{System model.}\vspace{-15pt}}
    \label{systemmodel}
\end{figure}
Denoting by $N$ the number of the sub-carriers, we transmit $G$ OFDM pilot signals with carrier frequency $f_c$ and bandwidth $B\ll f_c$. Given  the $g$-th pilot signal over the $n$-th sub-carrier $\boldsymbol{s}^{(g,n)}$ \footnote{The pilot signals are assumed to be known at the receiver and $\boldsymbol{s}^{(g,n)}\in\mathbb{C}^{N_t}$ is a general expression that permits the incorporation of the beamforming matrix, the design of which is beyond the scope of this paper.}, the $g$-th received signal over the $n$-th sub-carrier is given by \vspace{-5pt}
\begin{equation}
\boldsymbol{y}^{(g,n)}=\boldsymbol{H}^{(n)}  \boldsymbol{s}^{(g,n)}+\boldsymbol{w}^{(g,n)},\label{rsignal}\vspace{-5pt}
\end{equation}
where $\boldsymbol{w}^{(g,n)} \sim \mathcal{CN}(\bm{0},\sigma^2 \bm{I}_{N_r})$ is an independent, zero-mean, complex Gaussian vector with variance $\sigma^2$. We denote by $\bm{a}_N(f) \in \mathbb{C}^N$ the Fourier vector
\(
    \bm{a}_N(f) \triangleq \frac{1}{\sqrt{N}} \left[1, e^{-j 2\pi f}, \dots, e^{-j 2\pi (N-1)f}\right]^{\mathrm{T}}.
\)
The $n$-th sub-carrier channel matrix $\boldsymbol{H}^{(n)}$ with $0\leq n\leq N-1$ is then given by\vspace{-4pt}
\begin{equation}
\boldsymbol{H}^{(n)}\triangleq\sum_{k=0}^{K}\gamma_k e^{\frac{-j 2\pi n\tau_k}{N T_{s}}}\boldsymbol{\alpha}(\theta_{\mathrm{Rx},k})\boldsymbol{ \beta}\left(\theta_{\mathrm{Tx},k}\right)^{\mathrm{H}},
\label{channelmatrix_subcarrier}\vspace{-4pt}
\end{equation}
where $T_s\triangleq\frac{1}{B}$ is the sampling period, $\gamma_k\triangleq\sqrt{N_tN_r}\frac{h_{k}}{\sqrt{\rho_{k}}}$ is the channel coefficient of the $k$-th path, while $\rho_k$ and $h_k$ represents the path loss and the complex channel gain, respectively. The operator $[\cdot]^{\mathrm{H}}$ is the Hermitian transpose, and the steering vectors of the system, {\it i.e.,} $\boldsymbol{\alpha}(\theta_{\mathrm{Rx}})$ and $\boldsymbol{\beta}(\theta_{\mathrm{Tx}})$, are defined as $\boldsymbol{\alpha}(\theta_{\mathrm{Rx}})\triangleq \bm{a}_{N_r}\left(\frac{d \sin(\theta_{\mathrm{Rx}})}{\lambda_c} \right)$, $\boldsymbol{\beta}(\theta_{\mathrm{Tx}}) \triangleq \bm{a}_{N_t}\left(\frac{d \sin(\theta_{\mathrm{Tx}})}{\lambda_c} \right)$, where $d$ is the distance between antennas and $\lambda_c\triangleq \frac{c}{f_c}$ is the wavelength with $c$ being the speed of light. From the geometry shown in Fig. \ref{systemmodel}, the TOA, AOA, and AOD of each path, {\it i.e.,} $\tau_{k}$, $\theta_{\mathrm{Rx}, k}$, and $\theta_{\mathrm{Tx}, k}$, with $0\leq k\leq K$, are\vspace{-5pt}
\begin{subequations}\label{eq:geometricMapping}
\begin{align}
\tau_{0} &=\frac{\|\boldsymbol{p}-\boldsymbol{q}\|_{2}} { c},\label{tau0}\vspace{-4pt} \\
\tau_{k} &=\frac{\left\|\boldsymbol{q}-\boldsymbol{s}_{k}\right\|_{2} +\left\|\boldsymbol{p}-\boldsymbol{s}_{k}\right\|_{2}} {c}, \quad k>0, \vspace{-4pt}\\
\theta_{\mathrm{Tx}, 0} &=\arctan \left(\frac{p_{y}-q_{y} }{p_{x}-q_{x}}\right), \vspace{-3pt}\\
\theta_{\mathrm{Tx}, k} &=\arctan \left(\frac{s_{k,y}-q_{y}} {s_{k,x}-q_{x}}\right),  \quad k>0,\vspace{-3pt}\\
\vspace{-3pt}
\theta_{\mathrm{Rx}, 0} &=\pi+\arctan \left(\frac{p_{y}-q_{y}} {p_{x}-q_{x}}\right)-\theta_o,\vspace{-3pt}\\
\theta_{\mathrm{Rx}, k} &=\pi+\arctan \left(\frac{p_{y}-s_{k,y}}{p_{x}-s_{k,x}}\right)-\theta_o, \quad k>0,
\label{thetarxk}\vspace{-3pt}
\end{align}
\end{subequations}
where $k=0$ corresponds to the LOS path.
By stacking the received signal given in (\ref{rsignal}), we have\vspace{-3pt}
\begin{equation}
    \boldsymbol{Y} = \boldsymbol{H}\boldsymbol{S} + \boldsymbol{W};\vspace{-3pt}\label{signalmodelstacked} \vspace{-3pt}
\end{equation}
where $\boldsymbol{Y} \triangleq\left[\left(\boldsymbol{Y}^{(0)}\right)^{\mathrm{T}}, \left(\boldsymbol{Y}^{(1)}\right)^{\mathrm{T}}, \ldots, \left(\boldsymbol{Y}^{(N-1)}\right)^{\mathrm{T}}\right]^{\mathrm{T}}$, $\boldsymbol{H}\triangleq\operatorname{diag}\left\{\boldsymbol{H}^{(n)}\right\}$, $\boldsymbol{S} \triangleq\left[\left(\boldsymbol{S}^{(0)}\right)^{\mathrm{T}},   \left(\boldsymbol{S}^{(1)}\right)^{\mathrm{T}}, \ldots, \left(\boldsymbol{S}^{(N-1)}\right)^{\mathrm{T}}\right]^{\mathrm{T}}$, and $\boldsymbol{W} \triangleq\left[\left(\boldsymbol{W}^{(0)}\right)^{\mathrm{T}},   \left(\boldsymbol{W}^{(1)}\right)^{\mathrm{T}}, \ldots, \left(\boldsymbol{W}^{(N-1)}\right)^{\mathrm{T}}\right]^{\mathrm{T}}$. with
$\boldsymbol{Y}^{(n)} \triangleq \left[\boldsymbol{y}^{(1,n)}, \boldsymbol{y}^{(2,n)}, \ldots, \boldsymbol{y}^{(G,n)}\right]$, $\boldsymbol{S}^{(n)} \triangleq\left[ \boldsymbol{s}^{(1,n)},  \boldsymbol{s}^{(2,n)}, \right.$
$\left.\ldots,  \boldsymbol{s}^{(G,n)}\right]$,
and $\boldsymbol{W}^{(n)} \triangleq\left[ \boldsymbol{w}^{(1,n)},  \boldsymbol{w}^{(2,n)}, \ldots,  \boldsymbol{w}^{(G,n)}\right]$.

Furthermore, it is assumed that the receiver knows the transmitted symbols $\boldsymbol{S}$, and aims to estimate its orientation $\theta_o$ and target position $\bm{p}$.

\vspace{-10pt}
\section{Structure of mmWave MIMO OFDM narrowband channels}\label{sec:structure}
In this section, we present the structure of mmWave MIMO OFDM narrowband channel matrices. Without loss of generality, $N$ is assumed to be an odd integer, and we formulate a novel virtual channel matrix $\bm{H}_v$  to jointly exploit the received signal from all sub-carriers for localization as \vspace{-3pt}
\begin{equation}
\begin{aligned}
&\boldsymbol{H}_v\triangleq\sum_{k=0}^{K}l_k \left({\boldsymbol{\xi}}(\tau_k)\otimes\boldsymbol{\alpha}(\theta_{\mathrm{Rx},k})\right)\left({\boldsymbol{\xi}}(-\tau_k)\otimes\boldsymbol{ \beta}\left(\theta_{\mathrm{Tx},k}\right)\right)^{\mathrm{H}},
\end{aligned}\label{virtualchannel}\vspace{-3pt}
\end{equation}
where the operator $\otimes$ represents the Kronecker product, \mbox{$l_k \triangleq \frac{(N+1)\sqrt{N_tN_r}h_{k}}{2\sqrt{\rho_{k}}}$} and
${\boldsymbol{\xi}}(\tau)\triangleq \sqrt{\frac{2}{N+1}} \bm{a}_{\frac{N+1}{2}}(\frac{\tau}{N T_s})$.
We make some key observations regarding $\boldsymbol{H}_v$ in \eqref{virtualchannel},\vspace{-3pt}
\begin{itemize}
    \item[O1)] ${\boldsymbol{H}}_v$ is a low-rank matrix provided that $K+1\ll \min(N_r, N_t)< \min\left(\frac{(N+1)N_r}{2}, \frac{(N+1)N_r}{2}\right)$; \vspace{-6pt}
    \item[O2)] ${\boldsymbol{H}}_v$ has the same rank as $\bm{H}^{(n)}$, for any $n$; \vspace{-6pt}
    \item[O3)] ${\boldsymbol{H}}_v$ is a block Hankel matrix, {\it i.e.}, the ($i,j$)-th $N_r \times N_t$ block matrix ${\boldsymbol{H}}_v^{(i,j)}$ of ${\boldsymbol{H}}_v$ verifies ${\boldsymbol{H}}_v^{(i,j)}={\boldsymbol{H}}_v^{(k,z)}$ if $i+j=k+z$ for any $1\leq i,j,k,z\leq \frac{N+1}{2} $\vspace{-6pt}
    \item[O4)] ${\boldsymbol{H}}_v^{(i,j)}={\boldsymbol{H}}^{(i+j-2)}$ holds for any $1\leq i,j\leq \frac{N+1}{2} $, defining an automorphism $g$ between  ${\boldsymbol{H}}_v$ and $\boldsymbol{H}$ with $g(\boldsymbol{H}_v)=\boldsymbol{H}$.\vspace{-3pt}
\end{itemize}

\vspace{-10pt}
\section{Atomic norm based localization with orientation estimation}\label{sec:method} 

Based on the structural properties of the mmWave MIMO OFDM narrowband channel (O1-O4) discussed in Section \ref{sec:structure}, the low-rank property of each sub-carrier channel matrix can be ensured by exploiting the structure of the introduced virtual channel matrix. In this section, we first propose a tractable optimization problem to estimate the virtual channel matrix as well as the individual channel parameters based on ANM. Then, given the mappings (\ref{tau0})-(\ref{thetarxk}) between $\boldsymbol{\eta}\triangleq\{\tau_k,\theta_{\mathrm{Tx},k},\theta_{\mathrm{Rx},k}\}_{k \in \{0,1\ldots, K\}}$ and $\tilde{\boldsymbol{\eta}}\triangleq\{\boldsymbol q,\theta_o, \{\boldsymbol
s_k\}_{k \in \{0,1\ldots, K\}}\}$, we estimate the orientation and location of the target via a weighted least squares problem based on the EXIP \cite{Stoica}, which is compatible with our proposed ANM based channel estimator.\vspace{-3pt}
\subsection{Channel estimation}\label{sec:channelestimation}

From O1, $\boldsymbol{H}_v$ is a low-rank matrix. We harness the low-rank structure of $\boldsymbol{H}_v$ by defining the atomic set $\bm{\mathcal{A}}$ as \vspace{-5pt}
\begin{multline}\label{eq:atomicSet}
\boldsymbol{\mathcal{A}}\triangleq \big\{ \boldsymbol{A}\left(\tau,\theta_{\mathrm{Rx}}, \theta_{\mathrm{Tx}}\right) \triangleq\boldsymbol{\chi}(\tau,\theta_{\mathrm{Rx}})\boldsymbol{\zeta}(\tau,\theta_{\mathrm{Tx}})^{\mathrm{H}}|\\
\frac{d \sin \left(\theta_{\mathrm{Rx}}\right)}{\lambda_{c}}, {\frac{d \sin \left(\theta_{\mathrm{Tx}}\right)}{\lambda_{c}} }\in(-\frac{1}{2},\frac{1}{2}],\frac{\tau}{NT_s } \in(0,1].\big\}\vspace{-3pt}
\end{multline}
where ${{\boldsymbol{\chi}}}(\tau,\theta_{\mathrm{}})\triangleq{\boldsymbol{\xi}}(\tau)\otimes\boldsymbol{ \alpha}\left(\theta_{\mathrm{}}\right)
$ and ${{\boldsymbol{\zeta}}}(\tau,\theta_{\mathrm{}})\triangleq{\boldsymbol{\xi}}(-\tau)\otimes\boldsymbol{ \beta}\left(\theta_{\mathrm{}}\right).$
Proposition \ref{prop:equivalenceOfTheAtomicNorm} states the conditions under which  the atomic norm  $\left\Vert \cdot \right\Vert_{\mathcal{A}}$ induced by the atomic set  \eqref{eq:atomicSet}, defined as
\vspace{-3pt}$${\|\boldsymbol{H}_v \|}_{\mathcal{A}}\triangleq\inf\left\{\sum_{k}\left|{\tilde{l}_{k}}\right| \mid \boldsymbol{H}_v=\sum_{k} \tilde{l}_{k} \boldsymbol{A}\left(\tau_k,\theta_{\mathrm{Rx},k}, \theta_{\mathrm{Tx},k}\right)\right\}\vspace{-3pt}$$
can be calculated by solving a semidefinite program (SDP).
\begin{proposition}\label{prop:equivalenceOfTheAtomicNorm} Given the two conditions\vspace{-5pt}
\begin{itemize}
    \item [C1)]  $N_r, N_t\geq257$ and $N\geq513$\footnote{The condition C1 is just a technical requirement and is not necessary for the practical use \cite{Yang1,Tsai}.}, \vspace{-6pt}
    \item [C2)]
    $\Delta_{\min}\left(\frac{d \sin \left(\theta_{\mathrm{Rx}}\right)}{\lambda_c}\right)\geq\frac{1}{\lfloor\frac{N_r-1}{4}\rfloor}$, $\Delta_{\min}\left(\frac{d \sin \left(\theta_{\mathrm{Tx}}\right)}{\lambda_c}\right)\geq\frac{1}{\lfloor\frac{N_t-1}{4}\rfloor}$, and $\Delta_{\min}(\frac{\tau}{NT_s})\geq\frac{1}{\lfloor\frac{N-1}{8}\rfloor}$, where $\Delta_{\min}(\kappa)\triangleq\min_{i\neq j}\min(|\kappa_i-\kappa_j|,1-|\kappa_i-\kappa_j|)$\footnote{The condition C2 indicates that the bandwidth and the number of transmit and receive antennas need to be large enough to ensure that the TOAs, AODs, and AOAs of the $K+1$ paths, respectively, are sufficiently separated. Such conditions are necessary for the tightness of the ANM \cite{da2018tight}.}.\vspace{-3pt}
\end{itemize}
the atomic norm  ${\|\boldsymbol{H}_v \|}_{\mathcal{A}}$ is equivalently given by\vspace{-3pt}
\begin{equation}
\begin{aligned}
&\inf_{\boldsymbol{V},\boldsymbol{U},\boldsymbol{H}_v}\quad \frac{1}{2}\operatorname{Tr}\left(\boldsymbol{J}\right) \\
&\quad\text { s.t. }\boldsymbol{J}\triangleq\left[\begin{array}{cc}
\mathcal{T}_2({{\boldsymbol{U}}}) & {\boldsymbol{H}_v} \\
\boldsymbol{H}_v^{\mathrm{H}} & \mathcal{T}_2({{\boldsymbol{V}}})
\end{array}\right] \succeq \mathbf{0},\\
&\quad\quad\  \ \boldsymbol{H}_v^{(i,j)}=\boldsymbol{H}_v^{(k,z)},\textrm{ if } i+j=k+z, \forall i, j, k, z,
\end{aligned}\label{atomicnorm}\vspace{-3pt}
\end{equation}
where $\operatorname{Tr}(\cdot)$ denotes the trace of a matrix; $\mathcal{T}_2(\cdot)$ is a 2-level Toeplitz matrix constructed based on a matrix and its definition can be found in \cite{Yang2}. \label{equivalence}
\end{proposition}
\begin{proof}
Let $\|{\boldsymbol{H}_v}\|$ represent the objective value in (\ref{atomicnorm}) and define $\operatorname{SDP}(\boldsymbol{H}_v)$ according to \cite[Eq. (35)]{Tsai} as ,\vspace{-5pt}
\begin{equation}
\begin{aligned}
&\operatorname{SDP}(\boldsymbol{H}_v)\triangleq\inf_{\boldsymbol{V},\boldsymbol{U},\boldsymbol{H}_v}\quad \frac{1}{2}\operatorname{Tr}\left(\boldsymbol{J}\right) \\
&\quad\quad\quad\quad\quad\quad\ \text { s.t. }\boldsymbol{J}\triangleq\left[\begin{array}{cc}
\mathcal{T}_2({{\boldsymbol{U}}}) & {\boldsymbol{H}_v} \\
\boldsymbol{H}_v^{\mathrm{H}} & \mathcal{T}_2({{\boldsymbol{V}}})
\end{array}\right] \succeq \mathbf{0}.
\end{aligned}\label{SDP}\vspace{-3pt}
\end{equation}
The inequality $\|\boldsymbol{H}_v\| \geq \operatorname{SDP}(\boldsymbol{H}_v)$ holds based on the definitions of the key quantities. It can be shown from \cite[Lemma 1]{Tsai} that $\|\boldsymbol{H}_v\|\leq {\|\boldsymbol{H}_v\|}_{\mathcal{A}}$. Furthermore, from \cite[Lemma 2]{Tsai}  the equality  $\operatorname{SDP}(\boldsymbol{H}_v)={\|\boldsymbol{H}_v\|}_{\mathcal{A}}$ holds when conditions C1 and C2 are satisfied. Therefore, we conclude that $\|\boldsymbol{H}_v\|={\|\boldsymbol{H}_v\|}_{\mathcal{A}}$ given conditions C1 and C2.
\end{proof}
\begin{remark}
A Toeplitz-Hankel formulation is proposed in \cite{Cho} for the recovery of one-dimensional signals, which is proved to be equivalent to the atomic norm when the Hankel matrix therein admits a Vandermonde decomposition. Though the formulation in \cite{Cho} might be extended for the multi-dimensional case, the extended formulation still does not fit our signal model since $\boldsymbol{H}_v$ is not a Hankel matrix.
\end{remark}

From the observations (O1-O4) and Proposition \ref{equivalence}, the atomic norm denoiser of the virtual matrix $\bm{H}_v$ is formulated as, \vspace{-3pt}
\begin{align}
(\boldsymbol{\hat{{{U}}}}, \boldsymbol{\hat{{{ V}}}},\boldsymbol{\hat{{H}}}_v)={}&\mathop{\arg\min}_{\boldsymbol{V},\boldsymbol{U},\boldsymbol{H}_v}\quad \frac{\epsilon}{2}\operatorname{Tr}\left(\boldsymbol{J}\right)+\frac{1}{2}{\|\boldsymbol{Y}-g(\boldsymbol{H}_v)\boldsymbol{S}\|}^2_\mathrm{F} \nonumber \\
\text { s.t. }& \boldsymbol{J}\triangleq\left[\begin{array}{cc}
\mathcal{T}_2({{\boldsymbol{U}}}) & {\boldsymbol{H}_v} \nonumber\\
\boldsymbol{H}_v^{\mathrm{H}} & \mathcal{T}_2({{\boldsymbol{V}}})
\end{array}\right] \succeq \mathbf{0}, \nonumber\\
{}&  \boldsymbol{H}_v^{(i,j)}=\boldsymbol{H}_v^{(k,z)},\text{ if } i+j=k+z. \label{optprob}\vspace{-3pt}
\end{align}
where $\epsilon\varpropto\sigma\sqrt{(\frac{N+1}{2})^2N_rN_t\log((\frac{N+1}{2})^2N_rN_t)}$ (see \cite{Tsai}) is a regularization parameter and ${\|\cdot\|}_\mathrm{F}$ is the Frobenius norm. Furthermore, it is possible to estimate the TOAs, AODs and AOAs ($\hat\tau_k$, $\hat\theta_{\mathrm{Tx},k}$, and $\hat\theta_{\mathrm{Rx},k}$) from the Vandermonde decomposition of the solution $(\boldsymbol{\hat{{{U}}}}, \boldsymbol{\hat{{{ V}}}},\boldsymbol{\hat{{H}}}_v)$ of (\ref{optprob}). The estimated channel parameters corresponding to the same path are  paired via the matrix pencil and pairing algorithm~\cite{Yang2}. Note that, distinct from \cite{Beygi,Elnakeeb,Li}, a {\em multi-dimensional} ANM based estimator is proposed in this section and all the location-relevant parameters can be simultaneously recovered with super-resolution by harnessing the structure of $\bm{H}_v$.
\vspace{-10pt}
\subsection{Localization and orientation estimation}
\vspace{-5pt}
We assume that the LOS path is that with the smallest TOA. Though the estimated location and orientation can be directly computed based on the geometry of the LOS path, more accurate estimates can be achieved by exploiting the geometry of the NLOS paths \cite{Shahmansoori}.  Once the parameter $\boldsymbol{\eta}$, which parametrizes (\ref{optprob}) given the channel coefficients $\{a_k\}_{k \in \{0,1\ldots, K\}}$ is estimated through the procedure presented in Section \ref{sec:channelestimation}, the final step consists in recovering the location and orientation from \eqref{eq:geometricMapping}.

Since we make no assumptions on the path loss model in the signal model, knowledge of the channel coefficients do not improve the accuracy of the localization and orientation estimation. In addition, $\{\boldsymbol p,\theta_o, \{\boldsymbol{s}_k, a_k\}_{k \in \{0,1\ldots, K\}}\}$ can be used to re-parametrize the optimization problem in (\ref{optprob}). Therefore, we fix the estimated channel coefficients\footnote{We can substitute the estimated $\boldsymbol{\eta}$ into (\ref{signalmodelstacked}) to achieve a system of linear equations to compute the estimates of channel coefficients \cite{Li}.} and propose a weighted least squares problem to achieve an accurate localization and orientation estimation, with the estimates of all the paths, {\it i.e.,} $\hat{\boldsymbol\eta}$, exploited,\vspace{-5pt}
\begin{equation}
{\hat{\tilde{\boldsymbol\eta}}} = \arg\min_{\tilde{\boldsymbol\eta}} \left(\hat{\boldsymbol\eta}-f({\tilde{\boldsymbol\eta}})\right)^{\mathrm{T}}\mathcal{D}\left(\hat{\boldsymbol\eta}-f({\tilde{\boldsymbol\eta}})\right),
\label{exip_etatilde}\vspace{-5pt}
\end{equation}
where the mapping $f(\cdot)$ is defined according to the geometry, as described in (\ref{tau0})-(\ref{thetarxk}), with $f({\tilde{\boldsymbol\eta}})={\boldsymbol\eta}$. Inspired by the EXIP \cite{Stoica,Shahmansoori}, we denote by $L(\boldsymbol{\eta})$ the objective function in (\ref{optprob})  and use the Hessian matrix as the weight matrix in (\ref{exip_etatilde}), {\it i.e.,} \vspace{-6pt}
\begin{equation}
    \begin{aligned}
        \boldsymbol{\mathcal{D}} \triangleq
        &\left[\begin{array}{llll}
        \frac{\partial^2{L}(\hat{\boldsymbol\eta})}{\partial{\tau_1}\partial{\tau_1}}
        &\frac{\partial^2{L}(\hat{\boldsymbol\eta})}{\partial{\tau_1}\partial{\theta_{\mathrm{Tx},1}}} & \cdots&\frac{\partial^2{L}(\hat{\boldsymbol\eta})}{\partial{\tau_1}\partial{\theta_{\mathrm{Rx},L}}}\\
        \frac{\partial^2{L}(\hat{\boldsymbol\eta})}{\partial{\theta_{\mathrm{Tx},1}}\partial{\tau_1}}& \frac{\partial^2{L}(\hat{\boldsymbol\eta})}{\partial{\theta_{\mathrm{Tx},1}}\partial{\theta_{\mathrm{Tx},1}}}& \cdots & \frac{\partial^2{L}(\hat{\boldsymbol\eta})}{\partial{\theta_{\mathrm{Tx},1}}\partial{\theta_{\mathrm{Rx},L}}}\\
        & & \vdots & \\
        \frac{\partial^2{L}(\hat{\boldsymbol\eta})}{\partial{\theta_{\mathrm{Rx},L}}\partial{\tau_1}}& \frac{\partial^2{L}(\hat{\boldsymbol\eta})}{\partial{\theta_{\mathrm{Rx},L}}\partial{\theta_{\mathrm{Tx},1}}}& \cdots & \frac{\partial^2{L}(\hat{\boldsymbol\eta})}{\partial{\theta_{\mathrm{Rx},L}}\partial{\theta_{\mathrm{Rx},L}}}\\
        \end{array}\right],
    \end{aligned}\vspace{-3pt}
\end{equation}
which depends on the channel parameters estimated via the proposed ANM based method of Section \ref{sec:channelestimation}.

The non-linear least squares problem in (\ref{exip_etatilde}) can be solved via the Levenberg-Marquard-Fletcher algorithm \cite{Fletcher}. The parameters in $\tilde{\boldsymbol{\eta}}$ are initialized with the values $\boldsymbol{\hat p}_{\text{LOS}}$, $\hat{\theta}_{o,\text{LOS}}$, $\{\hat s_{k,y,\text{LOS}}\}_{k \in \{0,1\ldots, K\}}$, and $\{\hat s_{k,x,\text{LOS}}\}_{k \in \{0,1\ldots, K\}}$, which are derived in the following set of equations, \vspace{-7pt}
\begin{subequations}\label{eq:locEstimator}
\begin{align}
\boldsymbol{\hat p}_{\text{LOS}} &= \boldsymbol{q} + c\hat\tau_0[\cos({\hat\theta}_{\mathrm{Tx},0}),\sin(\hat\theta_{\mathrm{Tx},0})]^{\mathrm{T}}.\label{hatp}\vspace{-6pt}\\
\hat{\theta}_{_o,\text{LOS}} &= \pi + \hat\theta_{\mathrm{Tx},0} - \hat\theta_{\mathrm{Rx},0}.\\
\hat s_{k,y,\text{LOS}} &= \tan(\hat\theta_{\mathrm{Tx},k})(\hat s_{k,x}-q_{x})+q_{y}.
\label{hatsky}\vspace{-6pt}
\end{align}
\vspace{-15pt}
\begin{equation}
\begin{aligned}
&\hat s_{k,x,\text{LOS}}= \\
& \frac{\tan(\hat\theta_{\mathrm{Tx},k})q_{x}-\tan(\hat\theta_{\mathrm{Rx},k}+\hat{\theta}_{_o,\text{LOS}})\hat p_{\text{LOS},x}+\hat p_{\text{LOS},y}-q_{y}}{\tan(\hat\theta_{\mathrm{Tx},k})-\tan(\hat\theta_{\mathrm{Rx},k}+\hat{\theta}_{_o,\text{LOS}})}.
\label{hatskx}
\end{aligned}
\end{equation}

\end{subequations}
\vspace{-5pt}
\section{Numerical Results}\label{sec:sim}
\vspace{-5pt}
In this section, we evaluate the performance of our proposed scheme. In all of the numerical results, we set $f_c$, $B$, $c$, $N$, $N_r$, $N_t$, $G$, $K$, and $d$ to $60$ GHz, $100$ MHz, $300$ m/us, $15$, $16$, $16$, $16$, $2$, and $\frac{\lambda_c}{2}$, respectively. The BS is located at $[0 \text{ m},0 \text{ m}]^{\mathrm{T}}$ while the target is at $[20 \text{ m},5 \text{ m}]^{\mathrm{T}}$ with an orientation $\theta_o=0.2$ rad. The scatterers corresponding to two NLOS paths are placed at $[7.45 \text{ m}, 8.54 \text{ m}]^{\mathrm{T}}$ and $[19.89\text{ m}, -6.05 \text{ m}]^{\mathrm{T}}$, respectively. The channel coefficients are generated based on the free-space path loss model \cite{Goldsmith} in the simulation and the pilot signals are set as random complex values uniformly distributed on the unit circle. Note that condition C1 of Proposition \ref{equivalence} is generally not satisfied for the given $N$, $N_r$ and $N_t$ in our experiments; however, we still achieve strong performance.

\begin{figure}[htbp]
\centering\vspace{-4pt}
\includegraphics[scale=0.535]{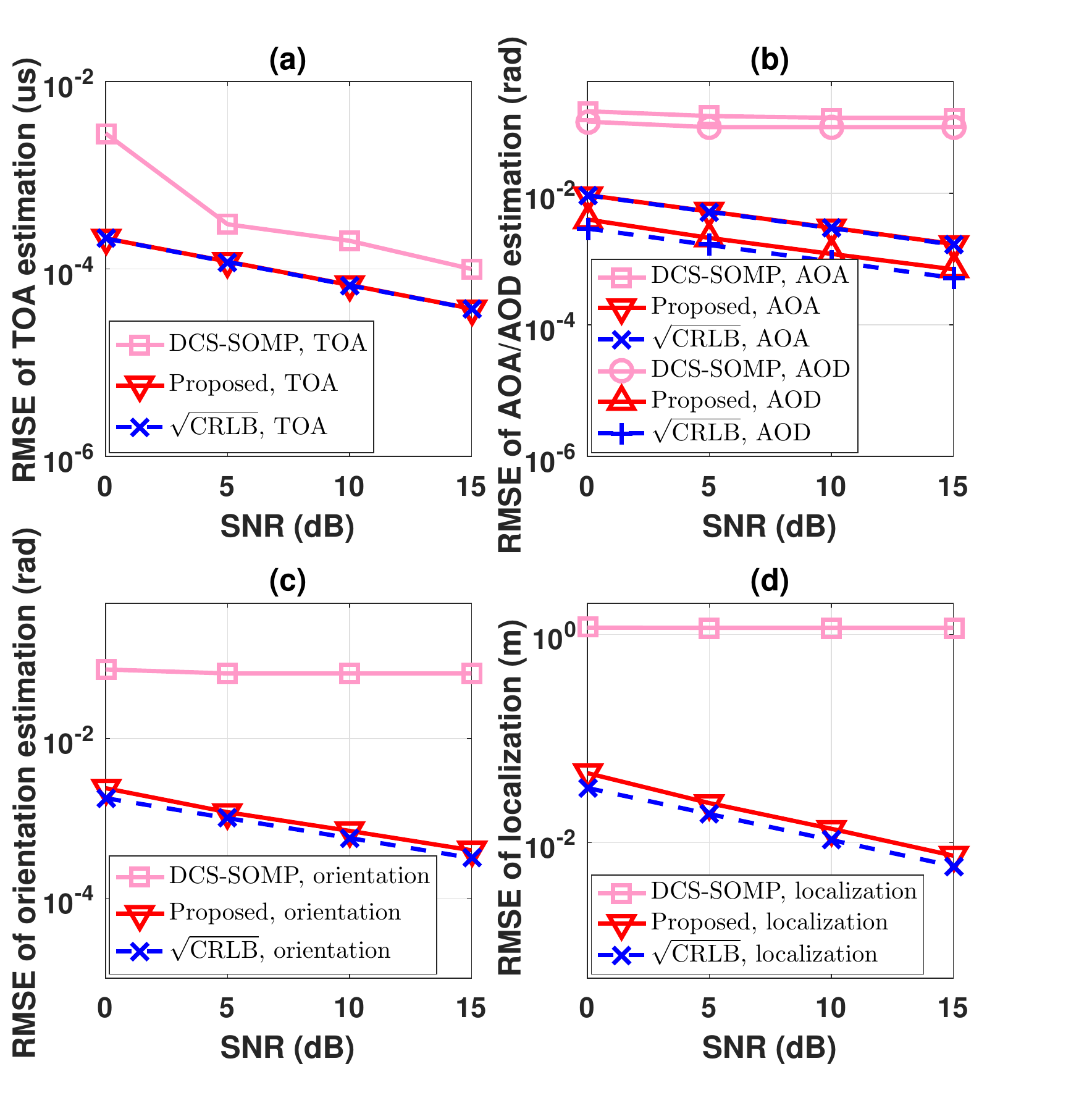}\vspace{-12pt}
\caption{{(a) The RMSE of TOA estimation; (b) The RMSE of AOA and AOD estimation; (c) The RMSE of orientation estimation; (d) The RMSE of localization.}\vspace{-10pt}}%
\label{numericalresults}
\end{figure}

The RMSEs of channel parameter estimation using our scheme are shown in Figs. \ref{numericalresults} (a) and (b), where the performance of DCS-SOMP \cite{Duarte,Shahmansoori} and CRLB \cite{Shahmansoori} are given as comparisons. As observed in Figs. \ref{numericalresults} (a) and (b), our proposed scheme outperforms DCS-SOMP, due to the fact that the grids of the AOAs and AODs are not dense enough for DCS-SOMP ({\it i.e.,} $N_t=N_r=16$); the estimation accuracy of our scheme does not rely on a grid resolution. Furthermore, the RMSEs of TOA, AOA, and AOD estimation using our scheme are close to or coincide with the corresponding CRLB curves according to Figs. \ref{numericalresults} (a) and (b).

Due to the quality of our super-resolution channel estimation,  lower RMSEs for localization and orientation estimation are achieved as seen in Figs. \ref{numericalresults} (c) and (d) versus the DCS-SOMP based method \cite{Shahmansoori}\footnote{To make a fair comparison, the refinement of estimates of channel parameters in \cite{Shahmansoori} is not implemented for both schemes. Note that, compared to DCS-SOMP, our scheme could provide more accurate estimates for the initialization of the refinement stage to avoid local optima.}. In addition, there is only around 2dB gap between the RMSE for localization or orientation estimation using our proposed scheme and the CRLB curves, verifying the efficacy of our design.

\vspace{-5pt}
\section{Conclusions}\label{sec:con}\vspace{-5pt}
In this paper, a multi-dimensional atomic norm based method is proposed for  high-accuracy localization and orientation estimation in mmWave MIMO OFDM systems.  To effectively estimate all of the location-relevant channel parameters with super-resolution, a novel virtual channel matrix is designed and its structure is fully exploited. Using the estimates of all the paths, a weighted least squares problem is proposed based on the extended invariance principle to accurately recover the location and orientation. The new method offers strong improvements with respect to the RMSE of estimation over prior work \cite{Shahmansoori} (more than $7$ dB gain). Furthermore, with the proposed method, the RMSEs of channel estimation, localization and orientation estimation are close to, or coincide with the corresponding CRLBs.

%


\vfill
\pagebreak



%
\renewcommand*{\bibfont}{\small}
\printbibliography


\end{document}